\documentclass[11pt]{article}

\newif\ifhyper\IfFileExists{hyperref.sty}{\hypertrue}{\hyperfalse}
\hypertrue
\ifhyper\usepackage{hyperref}\fi

\usepackage{enumitem}
\usepackage{soul}

\def\colorful{1}

\usepackage{amsthm,amsfonts,amsmath,amssymb,epsfig,color,float,graphicx,verbatim}
\usepackage{algorithm}
\usepackage[noend]{algpseudocode}
\usepackage{multirow}

\usepackage{enumitem}

\usepackage{framed}
\usepackage{nicefrac}

\ifnum\colorful=1
\newcommand{\new}[1]{{\color{red} #1}}

\newcommand{\toremove}[1]{\new{\st{#1}}}
\else
\newcommand{\new}[1]{{#1}}

\newcommand{\toremove}[1]{}
\fi

\newtheorem{theorem}{Theorem}[section]

\newtheorem{lemma}[theorem]{Lemma}

\newtheorem{claim}[theorem]{Claim}
\newtheorem{fact}[theorem]{Fact}

\theoremstyle{definition}

\newcommand{\R}{\mathbb{R}}

\newcommand{\E}{\mathbf{E}}
\newcommand{\Var}{\mathbf{Var}}
\newcommand{\Dom}{\mathbf{\Omega}}
\newcommand{\U}{\mathbf{u}}
\newcommand{\C}{\mathcal{C}_U}
\newcommand{\F}{\mathbf{F}}

\algnewcommand\INPUT{\item[{\textbf {input:}}]}
\algnewcommand\OUTPUT{\item[{\textbf{output:}}]}

\newcommand{\dtv}{d_{\mathrm TV}}
\newcommand{\wh}[1]{{\widehat{#1}}}

\newcommand{\ignore}[1]{}

\newcommand{\eps}{\epsilon}

\newcommand{\Poi}{\mathop{\textnormal{Poi}}\nolimits}

\newcommand{\eqdef}{\stackrel{{\mathrm {\footnotesize def}}}{=}}

\newcommand{\littlesum}{\mathop{\textstyle \sum}}

\floatstyle{ruled}
\newfloat{Algorithm}{H}{alg}
\newfloat{Subroutine}{H}{sub}

{
\begin{minipage}{1.0\linewidth} \begin{algorithm} %
} { \end{algorithm} \end{minipage} }

\title{Sharp Bounds for Generalized Uniformity Testing}

\author{
Ilias Diakonikolas\thanks{Supported by NSF Award CCF-1652862 (CAREER) and a Sloan Research Fellowship.}\\
University of Southern California\\
{\tt diakonik@usc.edu}\\
\and
Daniel M. Kane\thanks{Supported by NSF Award CCF-1553288 (CAREER) and a Sloan Research Fellowship.}\\
University of California, San Diego\\
{\tt dakane@cs.ucsd.edu}\\
\and
Alistair Stewart\\ University of Southern California\\
{\tt alistais@usc.edu}
}

\begin{document}

\maketitle

\begin{abstract}
We study the problem of {\em generalized uniformity testing}~\cite{BC17} of a discrete probability distribution:
Given samples from a probability distribution $p$ over an {\em unknown} discrete domain $\Dom$,
we want to distinguish, with probability at least $2/3$, between the case that $p$
is uniform on some {\em subset} of $\Dom$ versus $\eps$-far,
in total variation distance, from any such uniform distribution.

We establish tight bounds on the sample complexity of generalized uniformity testing.
In more detail, we present a computationally efficient tester
whose sample complexity is optimal, up to constant factors,
and a matching information-theoretic lower bound.
Specifically, we show that the sample complexity of generalized uniformity testing is
$\Theta\left(1/(\eps^{4/3}\|p\|_3) + 1/(\eps^{2} \|p\|_2) \right)$.
\end{abstract}

\thispagestyle{empty}
\setcounter{page}{1}


\section{Introduction}  \label{sec:intro}


Consider the following statistical task: Given independent samples from a distribution over an {\em unknown} discrete domain $\Dom$,
determine whether it is uniform on some {\em subset} of the domain versus significantly different from any such uniform distribution.
Formally, let $\C \eqdef \{\U_S: S \subseteq \Dom\}$ denote the set of uniform distributions $\U_S$ over subsets $S$ of $\Dom$.
Given sample access to an unknown distribution $p$ on $\Dom$ and a proximity parameter
$\eps>0$, we want to correctly distinguish between the case that $p \in \C$ versus $\dtv(p,  \C) \eqdef \min_{S \subseteq \Dom} \dtv(p, \U_S) \geq \eps$,
with probability at least $2/3$. Here, $\dtv(p, q) = (1/2) \|p-q\|_1$ denotes the total variation distance between distributions
$p$ and $q$. This natural problem, termed {\em generalized uniformity testing}, was recently studied by Batu and Canonne~\cite{BC17},
who gave the first upper and lower bounds on its sample complexity.

Generalized uniformity testing bears a strong resemblance to the familiar task of {\em uniformity testing},
where one is given samples from a distribution $p$ on an {\em explicitly known} domain of size $n$
and the goal is to determine, with probability at least $2/3$,
whether $p$ is the uniform distribution $\U_n$ on this domain versus $\dtv(p, \U_n) \geq \eps$.
Uniformity testing is arguably {\em the} most extensively studied problem in distribution property testing~\cite{GR00, Paninski:08, VV14, DKN:15,
Goldreich16, DiakonikolasGPP16, DGPP17} and its sample complexity is well understood.
Specifically, it is known~\cite{Paninski:08, CDVV14, VV14, DKN:15}
that $\Theta(n^{1/2}/\eps^2)$ samples are necessary and sufficient for this task.

The field of {\em distribution property testing}~\cite{BFR+:00} has seen substantial progress in the past decade,
see~\cite{Rub12, Canonne15} for two recent surveys.
A large body of the literature has focused on characterizing the sample size needed to test properties
of arbitrary distributions of a {\em given} support size. This regime is fairly well understood:
for many properties of interest there exist sample-efficient testers
~\cite{Paninski:08, CDVV14, VV14, DKN:15, ADK15, CDGR16, DK16, DiakonikolasGPP16, CDS17, DGPP17}.
Moreover, an emerging body of work has focused on leveraging a priori structure
of the underlying distributions to obtain significantly improved samples
complexities~\cite{BKR:04, DDSVV13, DKN:15, DKN:15:FOCS, CDKS17, DaskalakisP17, DaskalakisDK16, DKN17}.

Perhaps surprisingly, the natural setting where the distribution is arbitrary on a discrete but unknown domain (of unknown size)
does not seem to have been explicitly studied before the recent work of Batu and Canonne~\cite{BC17}.
Returning to the specific problem studied here, at first sight it might seem that generalized uniformity testing
and uniformity testing are essentially the same task. However, as shown in~\cite{BC17}, the sample complexities of these
two problems are significantly different. Specifically, \cite{BC17} gave a generalized uniformity tester with
expected sample complexity $O(1/(\eps^{6} \|p\|_{3}))$ and showed a lower bound of $\Omega(\|p\|_3)$.
Since generalized uniformity is a symmetric property, any tester should essentially rely on the empirical moments
(collision statistics) of the distribution~\cite{RRSS09, PV11sicomp}. The algorithm in~\cite{BC17} uses sufficiently accurate approximations of
the second and third moments of the unknown distribution. Their lower bound formalizes
the intuition that an approximation of the third norm is in some sense necessary to solve this problem.

\subsection{Our Results and Techniques} \label{ssec:techniques}

An immediate open question arising from the work of~\cite{BC17} is to precisely characterize the sample complexity
of generalized uniformity testing, as a function of all relevant parameters. The main result of this paper
provides an answer to this question. In particular, we show the following:

\begin{theorem}[Main Result] \label{thm:main}
There is an algorithm with the following performance guarantee: Given sample access to an arbitrary
distribution $p$ over an unknown discrete domain $\Dom$
and a parameter $0< \eps <1$, the algorithm uses $O\left(1/(\eps^{4/3}\|p\|_3) + 1/(\eps^{2} \|p\|_2) \right)$ 
independent samples from $p$ in expectation, and distinguishes between the case $p \in \C$ versus $\dtv(p, \C) \geq \eps$
with probability at least $2/3$. Moreover, 
for every $0< \eps < 1/0$ and $n >1$, any algorithm that distinguishes between $p \in \C$ 
and $\dtv(p, \C) \geq \eps$ requires at least $\Omega(n^{2/3}/\eps^{4/3}+n^{1/2}/\eps^2)$ samples, 
where $p$ is guaranteed to have $\|p\|_3 = \Theta(n^{-2/3})$ and $\|p\|_2=\Theta(n^{-1/2})$.
\end{theorem}

In the following paragraphs, we provide an intuitive explanation of our algorithm
and our matching sample size lower bound, in tandem with a comparison to the prior work~\cite{BC17}.

\paragraph{Sample-Optimal Generalized Uniformity Tester.}
Our algorithm requires considering two cases based on the relative
size of $\epsilon$ and $\|p\|_2^2$. This case analysis seems somewhat intrinsic to the problem
as the correct sample complexity branches into these cases.

For large $\epsilon$, we use the same overall technique as~\cite{BC17}, noting that
$p$ is uniform if and only if $\|p\|_3 = \|p\|_2^{4/3}$, and that for $p$ far
from uniform, $\|p\|_3$ must be substantially larger. The basic idea from
here is to first obtain rough approximations to $\|p\|_2$ and $\|p\|_3$ in
order to ascertain the correct number of samples to use, and then use
standard unbiased estimators of $\|p\|_2^2$ and $\|p\|_3^3$ to approximate
them to appropriate precision, so that their relative sizes can be
compared with appropriate accuracy.

We improve upon the work of~\cite{BC17}
in this parameter regime in a couple of ways. First, we obtain more precise
lower bounds on the difference $\|p\|_3^3 -\|p\|_2^4$ in the case where $p$ is far from uniform
(Lemma~\ref{lem:struct-unif}). This allows us to reduce the accuracy
needed in estimating $\|p\|_2$ and $\|p\|_3$. Second, we refine the
method used for performing the approximations to these moments ($\ell_r$-norms).
In particular, we observe that using the generic estimators for these quantities
yields sub-optimal bounds for the following reason:
The error of the unbiased estimators is related to their variance,
which in turn can be expressed in terms of the higher moments of $p$ (Fact~\ref{fact:collisions-moments}).
This implies for example that the worst case sample complexity
for estimating $\|p\|_3$ comes when the fourth and fifth moments of $p$ are large.
However, since we are trying to test for the case of uniformity (where these
higher moments are minimal), we do not need to worry about this worst case.
In particular, after applying sample efficient tests to ensure that the higher moments of $p$
are not much larger than expected (Lemma~\ref{lem:power-sum} (ii)), 
the standard estimators for the second and third moments
of $p$ can be shown to converge more rapidly than they would 
in the worst case (Lemma~\ref{lem:approx-F3-F2sq}).

The above algorithm is not sufficient for small values of $\epsilon$.
For $\eps$ sufficiently small, we employ a different, perhaps more
natural, algorithm. Here we take $m$ samples (for $m$ appropriately chosen
based on an approximation to $\|p\|_2$) and consider the subset $S$ of the domain
that appears in the sample. We then test whether the conditional distribution $p$ on $S$
is uniform, and output the answer of this tester.
The number of samples $m$ drawn in the first step is sufficiently large so that
$p(S)$, the probability mass of $S$ under $p$, is relatively high.
Hence, it is easy to sample from the conditional distribution using rejection sampling.
Furthermore, we can use a standard uniformity testing algorithm requiring $O(\sqrt{|S|}/\eps^2)$ samples.

To establish correctness of this algorithm, we need to show that if $p$ is far from uniform, then
the conditional distribution $p$ on $S$ is far from uniform as well. To prove this statement,
we distinguish two further subcases. If $\eps$ is ``very small'', then we can afford to set $m$
sufficiently large so that $p(S)$ is at least $1-\eps/10$. In this case, our claim follows straightforwardly.
For the remaining values of $\eps$, we can only guarantee that $p(S) = \Omega(1)$, hence we
require a more sophisticated argument. Specifically, we show (Lemma~\ref{lem:restriction}) 
that for any $x$ in an appropriate interval,
with high constant probability, the random variable $Z(x) = \sum_{i \in S} |p_i-x|$ is large.
It is not hard to show that this holds with high probability for each fixed $x$,
as $p$ being far from uniform implies that $\sum_{i \in \Dom} \min(p_i,|p_i-x|)$ is large.
This latter condition can be shown to provide a clean lower bound
for the expectation of  $Z(x)$. To conclude the argument,
we show that $Z(x)$ is tightly concentrated around its expectation.

\paragraph{Sample Complexity Lower Bound.}
The lower bound of $\Omega(1/(\eps^2 \|p\|_2))$ follows directly from the standard 
lower bound of $\Omega(n^{1/2}/\eps^2)$~\cite{Paninski:08} 
for uniformity testing on a given domain of size $n$. Specifically, it is implied from the fact
that the hard instances satisfy $\|p\|_2  = \Theta(n^{-1/2})$. The other branch of the lower bound, 
namely $\Omega(1/(\eps^{4/3} \|p\|_{3}))$, is more involved. To prove this lower bound, we use the shared information
method~\cite{DK16} for the following family of hard instances: In the ``YES'' case, 
we consider the distribution over (pseudo-)distributions on $N$ bins, where each $p_i$ is $(1 + \eps^2)/N$
with probability $n/(N(1+\eps^2))$, and $0$ otherwise. (Here we assume that the parameter 
$N$ is sufficiently large compared to the other parameters.)
In the ``NO'' case, we consider the distribution over (pseudo-)distributions on $N$ bins, where each $p_i$ is $(1 + \eps)/N$
with probability $n/(2N)$, $(1 - \eps)/N$ with probability $n/(2N)$, and $0$ otherwise.

\subsection{Notation} \label{ssec:notation}
Let $\Dom$ denote the unknown discrete domain.
Each probability distribution over $\Dom$ can be associated with
a probability mass function $p: \Dom \rightarrow \R_+$ such that $\sum_{i \in \Dom} p_i =1$.
We will use $p_i$, instead of $p(i)$, to denote the probability of element $i \in \Dom$ in $p$.
For a distribution (with mass function) $p$ and a set $S \subseteq \Dom$,
we denote by $p(S) \eqdef \sum_{i \in S} p_i$ and by
$(p|S)$ the conditional distribution of $p$ on $S$.
For $r \ge 1$, the $\ell_r$-norm of a function $p: \Dom \to \R$ is
$\|p\|_r \eqdef \left(\sum_{i \in \Dom} |p_i|^r\right)^{1/r}$.
For convenience, we will denote $\F_r(p) \eqdef  \|p\|_r^r  = \sum_{i \in \Dom} |p_i|^r$.
For $\emptyset \neq S \subseteq \Dom$, let $\U_S$ be the uniform distribution over $S$.
Let $\C  \eqdef \{\U_S: \emptyset \neq S \subseteq \Dom\}$
be the set of uniform distributions over subsets of $\Dom$.
The total variation distance between distributions $p, q$ on $\Dom$ is defined
as $\dtv(p, q) \eqdef \max_{S \subseteq \Dom} |p(S) - q(S)| = (1/2) \cdot \|p-q\|_1$.
Finally, we denote by $\mathrm{Poi}(\lambda)$ the Poisson distribution
with parameter $\lambda$.

\section{Generalized Uniformity Tester} \label{sec:upper}

In this section, we give our sample-optimal generalized uniformity tester,
\textsc{Gen-Uniformity-Test}. Before we describe our algorithm, we summarize a few preliminary
results on estimating the power sums $\F_r(p)  = \sum_{i \in \Dom} |p_i|^r$ of an unknown distribution $p$.
We present these results in Section~\ref{ssec:power-sum-estimation}.
In Section~\ref{ssec:pseudocode}, we give a detailed pseudo-code for our algorithm.
In Section~\ref{ssec:sample-comp}, we analyze the sample complexity, and in
Section~\ref{ssec:correctness} we provide the proof of correctness.

\subsection{Estimating the Power Sums of a Discrete Distribution} \label{ssec:power-sum-estimation}

We will require various notions of approximation for the power sums of a discrete distribution.
We start with the following fact:

\begin{fact}[\cite{AOST17}] \label{fact:collisions-moments}
Let $p$ be a probability distribution on an unknown discrete domain.
For any $r \geq 1$, there exists an estimator $\wh{\F}_r(p)$ for $\F_r(p)$
that draws $\Poi(m)$ samples from $p$ and satisfies the following:
$\E\left[\wh{\F}_r(p)\right] = \F_r(p)$  and $\Var\left[\wh{\F}_r(p) \right] = m^{-2r} \sum_{t=0}^{r-1} m^{r+t} \binom{r}{t} r^{r-t} \F_{r+t}(p)$.
\end{fact}

The estimator $\wh{\F}_r(p)$ is standard: It draws $\Poi(m)$ samples from $p$
and $m^r \cdot \wh{\F}_r(p)$ equals the number of $r$-wise collisions, i.e., ordered $r$-tuples
of samples that land in the same bin.
Using Fact~\ref{fact:collisions-moments}, we get
the following lemma which will be crucial for our generalized uniformity tester:
\begin{lemma} \label{lem:power-sum}
Let $p$ be a probability distribution on an unknown discrete domain and $r \geq 1$.
We have the following:
\begin{itemize}
\item[(i)] There exists an algorithm that, given a parameter $0< \delta <1$ and sample access to $p$,
draws $O(\frac{1}{\delta^2 \|p\|_r})$ samples from $p$ in expectation and outputs an estimate
$\wh{\gamma}_r$ that with probability at least $19/20$ satisfies: $|\wh{\gamma}_r - \F_r(p)| \leq \delta \cdot \F_r(p)$.

\item[(ii)] For any $c\geq 1$, there exist an algorithm that draws $\Poi\left(O(m)\right)$ samples from $p$ and
correctly distinguishes with probability at least $19/20$ between the case that
$m^r \F_r(p) \geq 20c$ versus $m^r \F_r(p) \leq c/20$.
\end{itemize}
\end{lemma}
\begin{proof}
Using Fact~\ref{fact:collisions-moments},
it is shown in~\cite{AOST17} that if we draw $m = O(\frac{1}{\delta^2 \|p\|_r})$ samples from $p$,
then with high constant probability we have that $|\wh{\F}_r(p) - \F_r(p)| \leq \delta \cdot \F_r(p)$.
Since the value of $\|p\|_r$ is unknown, this guarantee does not quite suffice for (i).
We instead start by approximating $1/\|p\|^r_r$ within a constant factor.
We do this by counting the number of samples we need to draw from $p$
until we see the first $r$-wise collision. By Fact~\ref{fact:collisions-moments} and
Chebyshev's inequality, this gives a constant factor approximation to $1/\|p\|^r_r$
with expected sample size of $O(1/\|p\|_r)$. We thus get (i).

We now proceed to show (ii). The algorithm is straightforward: Draw $\Poi\left(O(m)\right)$ samples from $p$ and
calculate $\wh{\F}_r(p)$. If $m^r \wh{\F}_r(p) > c$, output ``large''; otherwise output ``small''.
Suppose that $m^r \F_r(p) \leq c/20$. By Markov's inequality, with probability at least $19/20$ we will have that
$m^r  \wh{\F}_r(p) \leq c$, in which case we output ``small''.  Now suppose that $m^r \F_r(p) \geq 20c$.
Since $c \geq 1$, this gives that $\|p\|_r \geq 1/m$. Therefore, after we draw $\Poi(O(m))$ samples from $p$,
with probability at least $19/20$ we have that $\wh{\F}_r(p)$ is a factor $2$ approximation to $\F_r(p)$.
In other words, $m^r \wh{\F}_r(p) \geq 10c$ and the algorithm outputs ``large''.
\end{proof}

\subsection{Pseudo-code for \textsc{Gen-Uniformity-Test} Algorithm} \label{ssec:pseudocode}

The algorithm is given in the following pseudo-code:

\begin{algorithm}
\begin{algorithmic}[1]
\Procedure{Gen-Uniformity-Test}{$p, \eps$}

\INPUT  Sample access to arbitrary distribution $p$ on unknown discrete domain $\Dom$ and $\eps>0.$
\OUTPUT  ``YES'' with probability $2/3$ if $p \in \C$, ``NO'' with probability $2/3$ if $\dtv(p, \C) \geq \eps$.

\State  \label{step:F2:rough-estimation} Compute an estimate $\wh{\gamma}_2$ satisfying
           $|\wh{\gamma}_2 - \F_2(p)| \leq (1/2) \cdot \F_2(p)$ with probability $19/20$.

\State \label{step:n-defn} $n  \leftarrow  \lceil 2 / \gamma_2 \rceil$.

\If {($\eps \geq n^{-1/4})$}

\State  \label{step:F3:rough-estimation} Compute an estimate $\wh{\gamma}_3$ satisfying
           $|\wh{\gamma}_3 - \F_3(p)| \leq (1/2) \cdot \F_3(p)$ with probability $19/20$.

\If {($\wh{\gamma}_3 \geq 8/n^2$ or $\wh{\gamma}_3 \leq 1/(8n^2)$)} \label{step:check-F3} 
\textbf{return} ``NO''.
\EndIf

\State \label{sample-large-eps} Let $m \leftarrow \Theta(n^{2/3}/\eps^{4/3})$, for a sufficiently large constant in the $\Theta()$.

\State Let $c_4 = \Theta(1+m^4/n^3)$, for a sufficiently large constant in the $\Theta()$.
\State \label{step:F4-samples} Draw $\Poi(O(m))$ samples from $p$ and let $\wh{\gamma}_4$ denote the value of $\wh{\F}_4(p)$ on this sample.
\If {$m^4 \wh{\gamma}_4 > 20c_4$} \label{step:F4:upper-bound}
\textbf{return} ``NO''.
\EndIf

\State Let $c_5 = \Theta(1+m^5/n^4)$, for a sufficiently large constant in the $\Theta()$.
\State  \label{step:F5-samples} Draw $\Poi(O(m))$ samples from $p$ and let $\wh{\gamma}_5$ denote the value of $\wh{\F}_5(p)$ on this sample.
\If {$m^5 \wh{\gamma}_5 > 20c_5$} \label{step:F5:upper-bound}  \textbf{return} ``NO''.
\EndIf

\State  \label{step:F2-F3-strong-approx} Compute the estimates $\wh{\F}_2(p)$, $\wh{\F}_3(p)$
           on two separate sets of $\Poi(m)$ samples.

\If {$\left(\wh{\F}_3(p) - \wh{\F}_2(p)^2 \leq \eps^2/(300n^2) \right)$}
\textbf{return} ``YES''.
\Else  { \textbf{return} ``NO''.}
\EndIf
\EndIf

\If {($n^{-1/4} \log^{-1}(n) \leq \eps < n^{-1/4}$)}
     \State Let $m_1 \leftarrow \Theta(n)$, for an appropriately large constant in the $\Theta()$.

\State \label{step:find-s} Draw $\Poi(m_1)$ samples from $p$. Let $S$ be the subset of $\Dom$ that appears in the sample.

\State \label{step:no-heavy} Verify the following conditions: (i) Each $i \in S$ appears $O(\log n)$ times;
\State  \label{step:support-mass} (ii) $|S| \geq n/2$; (iii) $p(S) \geq 1/2$.
\If {(any of conditions~(\ref{step:no-heavy}),~(\ref{step:support-mass}) is violated)}
\textbf{return} ``NO''. \label{step:violate}
\EndIf
\State Using rejection sampling, draw $m_2 \leftarrow O(n^{1/2}/\eps^2)$ samples from $(p|S)$.
\State Test whether $(p|S) = \U_S$ versus $\eps/10$-far from $\U_S$
with confidence probability $19/20$. \label{step:uniform-restriction}
\State \textbf{return} the answer of the tester in Step~\ref{step:uniform-restriction}.
\EndIf

\If {($\eps < n^{-1/4} \log^{-1}(n)$)}
\State $m_1 \leftarrow \Theta(n \log n)$, for an appropriately large constant in the $\Theta()$.
     \State \label{step:find-S} Draw $\Poi(m_1)$ samples from $p$. Let $S$ be the subset of $\Dom$ that appears in the sample.
\State Draw $m_2 \leftarrow O(n^{1/2}/\eps^2)$ samples from $p$.
\If {(any of the above samples lands outside $S$)} \label{step:test-small-eps}
\textbf{return} ``NO''.
\EndIf
\State Test whether $(p|S) = \U_S$ versus $\eps/2$-far from $\U_S$
with confidence probability $19/20$. \label{step:uniform-restriction2}
\State \textbf{return} the answer of the tester in Step~\ref{step:uniform-restriction2}.
\EndIf

\EndProcedure
\end{algorithmic}
\caption{Sample-Optimal Algorithm for Generalized Uniformity Testing}
\label{alg:gen-unif-test}
\end{algorithm}


\subsection{Bounding the Sample Complexity} \label{ssec:sample-comp}
We start by analyzing the sample complexity of the algorithm.
We claim that the expected sample complexity is 
$O\left(1/\big(\eps^{4/3} \|p\|_3\big) \right)$ for $\eps \geq n^{-1/4}$ and
$O\left(1/\big(\eps^{2} \|p\|_2\big) \right)$ for $\eps < n^{-1/4}$.

By Lemma~\ref{lem:power-sum} (i), 
Step~\ref{step:F2:rough-estimation} can be implemented with expected sample complexity $O(1/\|p\|_2)$
and Step \ref{step:F3:rough-estimation} with expected sample complexity $O(1/\|p\|_3)$.

We start with the case $\eps \geq n^{-1/4}$.
If Steps~\ref{step:F2:rough-estimation},~\ref{step:F3:rough-estimation}, and~\ref{step:check-F3} succeed,
then we have that $\F_2(p) = \Theta(1/n)$ and $\F_3(p) = \Theta(1/n^2)$. Also note that no further steps
are executed unless the condition of Step ~\ref{step:check-F3} holds.
Note that all subsequent steps that draw samples 
(Steps~\ref{step:F4-samples},~\ref{step:F5-samples}, and~\ref{step:F2-F3-strong-approx}) 
by definition use at most $\Poi(O(m))$ additional samples.
Since Step~\ref{step:F2-F3-strong-approx} is executed only if $\wh{\gamma}_3 = \Theta(1/n^2)$, 
we have that $m = O(\wh{\gamma}_3^{-1/3}/\eps^{4/3}) = O(1/(\eps^{4/3}\|p\|_3))$.
Therefore, for $\eps \geq n^{-1/4}$,
the expected sample complexity of the algorithm is bounded by
$$O\left(1/\|p\|_2\right) + O\left(1/\|p\|_3\right) + O\left(1/\big(\eps^{4/3} \|p\|_3\big) \right) = O\left(1/\big(\eps^{4/3} \|p\|_3\big) \right) \;.$$

For the case of $n^{-1/4} \log^{-1} (n) \leq \eps < n^{-1/4}$, the additional sample size drawn on top of Step~\ref{step:F2:rough-estimation}
is $O(n +n^{1/2}/\eps^2) = O(n^{1/2}/\eps^2)$. Since $n = \Theta (1/\|p\|^2_2)$, the total sample complexity in this case is
$$O\left(1/\|p\|_2\right) +  O\left(1/\big(\eps^{2} \|p\|_2\big) \right) = O\left(1/\big(\eps^{2} \|p\|_2\big) \right) \;.$$
Finally, for  $\eps < n^{-1/4} \log^{-1} (n)$, the sample size drawn on top of Step~\ref{step:F2:rough-estimation}
is $O(n \log n +n^{1/2}/\eps^2) = O(n^{1/2}/\eps^2)$. Since $n = \Theta (1/\|p\|^2_2)$, the total sample complexity in this case is
$O\left(1/\big(\eps^{2} \|p\|_2\big) \right)$, as before. This completes the analysis of the sample complexity.


\subsection{Correctness Proof} \label{ssec:correctness}

This section is devoted to the correctness proof of \textsc{Gen-Uniformity-Test}.
In particular, we will show that if $p \in \C$, the algorithm outputs ``YES''
with probability at least $2/3$ (completeness); and if $\dtv(p, \C) \geq \eps$,
the algorithm outputs ``NO'' with probability at least $2/3$ (soundness).

We start with the following simple claim giving
a useful condition for the soundness case:
\begin{claim} \label{claim:far-uniform}
If  $\dtv(p, \C) \geq \eps$, then for all $x \in [0, 1]$ we
have that $\sum_{i \in \Dom} \min\{ p_i, |x - p_i| \}  \geq \eps/2.$
\end{claim}
\begin{proof}
Let $S_{h}$ be the set of $i \in \Dom$ on which $p_i>x/2$. 
Let $\delta = \sum_{i \in \Dom} \min\{ p_i, |x - p_i| \}.$ 
Note that $\delta=\|p-c_{x,S_h}\|_1$, 
where $c_{x,S_h}$ is the pseudo-distribution 
that is $x$ on $S_h$ on $0$ elsewhere. 
If $\|c_{x,S_h}\|_1$ were $1$, $c_{x,S_h}$ would be the uniform distribution 
$\U_{S_h}$ and we would have $\delta\geq \eps$. 
However, this need not be the case. 
That said, it is easy to see that 
$\|\U_{S_h}-c_{x,S_h}\|_1 = |1-\|c_{x,S_h}\|_1| \leq \|p-c_{x,S_h}\|_1 = \delta$. 
Therefore, by the triangle inequality
$$
2\delta \geq \|p-c_{x,S_h}\|_1 + \|\U_{S_h}-c_{x,S_h}\|_1 \geq \|p-\U_{S_h}\|_1 \geq \eps \;.
$$
This completes the proof of Claim~\ref{claim:far-uniform}.
\end{proof}

\noindent We now proceed to analyze correctness for the various ranges of $\eps.$

\medskip

\noindent {\bf Case I:  [$\eps \geq n^{-1/4}$]}. 
The following structural lemma provides a reformulation of generalized uniformity testing
in terms of the second and third norms of the unknown distribution:

\begin{lemma} \label{lem:struct-unif}
We have the following:
\begin{itemize}
\item[(i)] If $p \in \C$, then $\F_3(p) = \F_2^2(p)$.
\item[(ii)] If $\dtv(p, \C) \geq \eps$, then $\F_3(p) - \F_2^2(p) > \eps^2 \F_2^2(p)/64$.
\end{itemize}
\end{lemma}
\begin{proof}
The proof of (i) is straightforward. Suppose that $p = \U_S$ for some $\emptyset \neq S \subseteq \Dom$.
It then follows that $\F_2(p) = 1/|S|$ and $\F_3(p) = 1/|S|^2$, yielding part (i) of the lemma.

We now proceed to prove part (ii).
Suppose that $\dtv(p, \C) \geq \eps$.
First, it will be useful to rewrite the quantity $\F_3(p) - \F_2^2(p)$ as follows:
\begin{equation} \label{eqn:F3-F2sq}
\F_3(p) - \F_2^2(p) = \sum_{i \in \Dom} p_i (p_i - \F_2(p))^2 \;.
\end{equation}
Note that (\ref{eqn:F3-F2sq}) follows from the identity $ p_i (p_i - \F_2(p))^2 = p_i^3 + p_i \F_2(p)^2  - 2p_i^2 \F_2(p)$
by summing over $i \in \Dom$.
Since $\dtv(p, \C) \geq \eps$, an application of Claim~\ref{claim:far-uniform} for $x = \F_2(p) \in [0, 1]$, gives that
$$\sum_{i \in \Dom} \min\{ p_i, |\F_2(p) - p_i| \} \geq \eps/2 \;.$$
We partition $\Dom$ into the sets $S_{l} = \{ i \in \Dom \mid p_i < \F_2(p)/2\}$
and its complement $S_{h} = \Dom \setminus S_{l}$. Note that
$\sum_{i \in \Dom} \min\{ p_i, |\F_2(p) - p_i| \} = \sum_{i \in S_{l}} p_i + \sum_{i \in S_h} |\F_2(p) - p_i| \;.$
It follows that either $\sum_{i \in S_{l}} p_i  \geq \eps/4$
or  $\sum_{i \in S_{h}} |\F_2(p)-p_i| \geq \eps/4$. We analyze each case separately.
First, suppose that $\sum_{i \in S_{l}} p_i  \geq \eps/4$. Using (\ref{eqn:F3-F2sq}) we can now write
$$\F_3(p) - \F_2^2(p) \geq \sum_{i \in S_{l}} p_i (p_i - \F_2(p))^2 > (\F_2(p)/2)^2 \cdot \sum_{i \in S_{l}} p_i  = \eps \F^2_2(p) /16  \;.$$
Now suppose that $\sum_{i \in S_{h}} |\F_2(p)-p_i| \geq \eps/4$. Note that $1 \leq |S_h| \leq 2/|\F_2(p)|$.
In this case, using (\ref{eqn:F3-F2sq}) we obtain
\begin{eqnarray*}
\F_3(p) - \F_2^2(p)
&\geq& \sum_{i \in S_{h}} p_i (p_i - \F_2(p))^2 \\
&\geq& (\F_2(p)/2) \cdot \sum_{i \in S_{h}} (p_i - \F_2(p))^2 \\
&\geq&  (\F_2(p)/2) \cdot \frac{(\sum_{i \in S_{h}} |\F_2(p)-p_i|)^2}{|S_h|} \\
&\geq&  (\F_2(p)/2)^2 \cdot (\eps/4)^2\\
&=&   \eps^2  \F^2_2(p) / 64 \;,
\end{eqnarray*}
where the second inequality uses the definition of $S_h$,
and the third inequality is Cauchy-Schwarz.
This completes the proof of Lemma~\ref{lem:struct-unif}.
\end{proof}


By Lemma~\ref{lem:struct-unif},
the proof in this case boils down to proving that our estimates for $\F_2(p)$
and $\F_3(p)$ obtained in Step~\ref{step:F2-F3-strong-approx} are sufficiently accurate
to distinguish between the completeness and soundness cases.
We note that since Steps~(\ref{step:check-F3}), (\ref{step:F4:upper-bound}), and
(\ref{step:F5:upper-bound}) have succeeded, with probability at least $19/20$ each of the corresponding conditions
is satisfied. Specifically, this implies that the following conditions hold: 
$\F_2(p) = \Theta (1/n)$, $\F_3(p) = \Theta (1/n^2)$, $\F_4(p) = O(m^{-4} + n^{-3})$, and 
$\F_5(p) = O(m^{-5} + n^{-4})$.

We henceforth condition on this event. The following lemma shows that our approximations to the second
and third moments are appropriately accurate:

\begin{lemma} \label{lem:approx-F3-F2sq}
Let $c$ be an appropriately small universal constant (selecting $c = 10^{-3}$ suffices).
With probability at least $9/10$ over the samples, the estimates for $\F_2(p)$
and $\F_3(p)$ obtained in Step~\ref{step:F2-F3-strong-approx} satisfy the following conditions:
\begin{itemize}
\item[(i)] $|\wh{\F_2}(p) - \F_2(p)| \leq  c \cdot \eps^2 \F_2(p)$.
\item[(ii)] $|\wh{\F_3}(p) - \F_3(p)| \leq  c \cdot \eps^2 \F_2^2(p)$.
\end{itemize}
\end{lemma}
\begin{proof}
The lemma follows using Fact~\ref{fact:collisions-moments} and an application of Chebyshev's inequality, 
crucially exploiting the improved variance bounds that hold when the above
conditions are satisfied.


To prove part (i), note that $\Var[\wh{\F}_2(p)] = O\left( m^{-2} \F_2(p) + m^{-1} \F_3(p)  \right)$.
We use that $\F_3(p) = \Theta (1/n^2) = \Theta (\F^2_2(p))$,
where the second inequality uses the fact that $1/n = \Theta(\F_2(p))$
(as follows from Steps~\ref{step:F2:rough-estimation} and~\ref{step:n-defn} of the algorithm).
Now recall that the sample size $m$ is defined to be $\Theta(n^{2/3}/\eps^{4/3})$, for a sufficiently
large universal constant in the big-$\Theta$.
We can therefore bound the variance $\Var[\wh{\F}_2(p)]$ from above by
$$O\left( m^{-2} n^{-1} + m^{-1} n^{-2}  \right) =  O\left( \eps^{8/3} n^{-7/3} + \eps^{4/3} n^{-8/3}   \right) = O(\eps^4/n^2) \;,$$
where we used the assumption that $\eps \geq n^{-1/4}$. By Chebyshev's inequality, we therefore get that
\begin{equation} \label{eqn:chebyshev:F2}
|\wh{\F_2}(p) - \F_2(p)| \leq O(\eps^2/n) \;,
\end{equation}
with probability at least $19/20$. By selecting the constant factor 
in the definition of $m$ appropriately, we can make the RHS in (\ref{eqn:chebyshev:F2})
at most $c \cdot \eps^2 \F_2(p)$, as desired.

Part (ii) is proved similarly. We have that
$\Var[\wh{\F}_3(p)] = O\left( m^{-3} \F_3(p) + m^{-2} \F_4(p) +  m^{-1} \F_5(p)  \right)$.
We use that $\F_3(p) = \Theta (1/n^2)$,
$\F_4(p) = O(m^{-4} + n^{-3})$, and $\F_5(p) = O(m^{-5} + n^{-4})$.
Recalling that the sample size $m$ is defined to be $\Theta(n^{2/3}/\eps^{4/3})$,
we can bound the variance $\Var[\wh{\F}_3(p)]$ from above by
$$O\left( m^{-3} n^{-2} + m^{-6} + m^{-2} n^{-3} +  m^{-1} n^{-4} \right) =  O\left(  \eps^4/n^4  \right) \;,$$
where we used the assumption that $m = \Theta(n^{2/3}/\eps^{4/3})$ and $\eps \geq n^{-1/4}$.
By Chebyshev's inequality, we therefore get that
\begin{equation} \label{eqn:chebyshev:F3}
|\wh{\F_3}(p) - \F_3(p)| \leq O(\eps^2/n^2) \;,
\end{equation}
with probability at least $19/20$. By selecting the constant in the big-$\Theta$
defining $m$ appropriately, it is clear that we can make the RHS in (\ref{eqn:chebyshev:F3})
at most $c \cdot \eps^2 \F_2^2(p)$, as desired.
This completes the proof of Lemma~\ref{lem:approx-F3-F2sq}.
\end{proof}

We now have all the necessary ingredients to establish completeness and soundness in Case I.
If $p \in \C$, it is easy to see that Steps~(\ref{step:check-F3}), (\ref{step:F4:upper-bound}), and
(\ref{step:F5:upper-bound}) succeed with high constant probability, as follows from the fact
that the norms are minimal in this case and Lemma~\ref{lem:power-sum}. Moreover, if the algorithm 
does not reject in any of these steps, the corresponding conditions on the magnitude of these
norms are satisfied. If the conditions of Lemma~\ref{lem:approx-F3-F2sq} hold, then we have that
$$\left| \left(\F_3(p) - \F_2^2(p)\right) - \left( \wh{\F}_3(p) - \wh{\F}_2(p)^2  \right)\right|  \leq c \cdot \eps^{2} \F^2_2(p) \;.$$
Therefore, the algorithm correctly distinguishes between the completeness
and soundness cases, via Lemma~\ref{lem:struct-unif}.
This completes the correctness analysis of Case I.

\bigskip\medskip

\noindent {\bf Case II:  [$n^{-1/4} \log^{-1}(n) \leq \eps < n^{-1/4}$]}. 
The correctness in the completeness case is straightforward.
If $p \in \C$, it is easy to see that Conditions~\ref{step:no-heavy} and~\ref{step:support-mass}
will be satisfied with high constant probability.
Moreover, the conditional distribution $(p|S)$ equals $\U_S$,
and therefore the overall algorithm outputs  ``YES'' with high constant probability.

The correctness of the soundness case is more involved. 
Suppose that $\dtv(p, \C) \geq \eps$.
If the algorithm does not output ``NO'' in Step~\ref{step:violate}, the following conditions hold with high probability:
(a) $|S| \geq n/2$, (b) $p(S) \geq 1/2$, and (c) $p_i  = O(\log n / n)$ for all $i \in \Dom$.
We will use these statements to prove the following lemma:

\begin{lemma} \label{lem:restriction}
If $\dtv(p, \C) \geq \eps$ and the conditions in Steps~\ref{step:no-heavy}, \ref{step:support-mass} hold, then
with high constant probability over the samples drawn in Step~\ref{step:find-s}, we have that $\dtv((p|S), \U_S) \geq \eps/10$.
\end{lemma}
\begin{proof}
Suppose that $\dtv(p, \C) \geq \eps$. We want to show that with high probability over the samples it holds
$\sum_{i \in S} \left| p_i - p(S) / |S|\right|  = \Omega(\eps)$.
The main difficulty is that the value of $p(S)$ is unknown, hence
we need a somewhat indirect argument.
By Claim~\ref{claim:far-uniform}, for all $x \in [0, 1]$ we have that
\begin{equation} \label{eqn:far-uniform2}
\sum_{i \in \Dom} \min\{ p_i,  |p_i- x| \} \geq \eps/2 \;.
\end{equation}
To show that $\sum_{i \in S} \left| p_i - p(S) / |S|\right|  = \Omega(\eps)$,
it suffices to prove that the following holds:
\begin{claim} \label{claim:x}
With probability at least $19/20$,
for all $x$ in an additive grid with step size $O(\eps/n)$ such that $0 \leq x \leq \log n / n$,  we have that
$Z(x) \eqdef \sum_{i \in S} |p_i - x|  = \Omega(\eps)$.
\end{claim}
First note that for $x > 4/n$ or $x<1/(4n)$, the above claim is satisfied automatically.
Indeed, for $x>4/n$, we have $\sum_{i \in S} |p_i - x| \geq |S| \cdot x - p(S) \geq (n/2)x-1 \geq 1$.
For $x < 1/(4n)$, we have $\sum_{i \in S} |p_i - x| \geq p(S) - |S| \cdot x \geq 1/2 - n x \geq 1/4$.

We henceforth focus on the setting where $1/(4n) \leq x \leq 4/n$.
Here we show that $\E[Z(x)]$ is large and that $Z$ is tightly concentrated around its expectation.

Let $Z_i$, $i \in \Dom$, be the indicator
of the event $i \in S$. Then, $Z(x) = \sum_{i \in \Dom}  |p_i - x| Z_i$. Note that $Z_i$ is a Bernoulli random variable
with $\E[Z_i] =1 - e^{-p_i n}$ and that the $Z_i$'s are mutually independent.
Note that
$\E[Z(x)] = \sum_{i \in \Dom} (1 - e^{-p_i n}) |p_i- x|$.
We recall the following concentration inequality for sums of non-negative random variables
(see, e.g., Exercise 2.9 in~\cite{Boucheron13}):
\begin{fact} \label{thm:chb}
Let $X_1, \ldots, X_m$ be independent non-negative random variables,
and $X \ = \littlesum_{j=1}^{m} X_j$.
Then, for any $t>0$, it holds that
$\Pr [ X \leq \E[X] - t] \le \exp \left( -t^2 /(2 \sum_{i=1}^m \E[X_i^2]) \right).$
\end{fact}
\noindent Since $Z(x) = \sum_{i \in \Dom}  |p_i - x| Z_i$ where the $Z_i$'s are independent Bernoulli random variables
with $\E[Z_i^2] = 1 - e^{-p_i n}$,
an application of Fact~\ref{thm:chb} yields that
\begin{equation} \label{eqn:chernoff}
\Pr\left[Z(x) \leq \E[Z(x)] -t\right]  \leq \exp \left( \frac{-t^2}{2 \littlesum_{i \in \Dom} (1 - e^{-p_i n})(p_i-x)^2}  \;. \right)
\end{equation}
Let $S_{l} =  \{ i \in \Dom: p_i \leq x/2 \}$ and $S_h = \Dom \setminus S_{l}$.
By  (\ref{eqn:far-uniform2}), we get that
$\sum_{i \in S_{l}} p_i + \sum_{i \in S_{h}} |x-p_i|  \geq \eps/2 \;.$
For $i \in S_l$, we have that $(1 - e^{-p_i n}) |p_i- x| \geq n \cdot p_i \cdot |x/2| = \Omega(p_i)$.
For $i \in S_h$, we have that $(1 - e^{-p_i n}) = \Omega(1)$ and therefore
$(1 - e^{-p_i n}) |p_i- x| = \Omega(1) |p_i- x|.$ We therefore get that
$\E[Z(x)] = \Omega(\eps)$.
We now bound $\littlesum_{i \in \Dom} (1 - e^{-p_i n})(p_i-x)^2$ from above
using the fact that $p_i \leq \log n / n$, for all $i \in \Omega$.
This assumption and the range of $x$ imply that
$$\littlesum_{i \in \Dom} (1 - e^{-p_i n})(p_i-x)^2 \leq O(\log n / n) \cdot \E[Z] \;.$$
So, by setting $t = \E[Z]/2$ in (\ref{eqn:chernoff}), we get that
$$ \Pr[Z(x) \leq \E[Z(x)]/2]  \leq \exp \left( -\Omega(\eps n / \log n)\right) = \exp\left(-n^{\Omega(1)}\right) \;,$$
where the last inequality follows from the range of $\eps$.
Recalling that $x$ lies in a grid of size $O(n/\eps)$, Claim~\ref{claim:x} follows by a union bound.
This completes the analysis of Case II.

\medskip

\noindent {\bf Case III:  [$\eps < n^{-1/4} \log^{-1}(n) $]}.
The correctness in this case is quite simple.
In the completeness case, conditioning on Step~\ref{step:F2:rough-estimation} succeeding,
we know that $p$ is uniform over a domain of size $O(n)$. Therefore, after $\Theta(n \log n)$ samples, 
we have seen all the elements of the domain with high probability, i.e., the set $S$ 
has $p(S) = 1$. Therefore, the conditional distribution $p|S$ is identified
with $p$, and the final tester outputs ``YES''. 
On the other hand, if $p$ is $\eps$-far from uniform. and the algorithm does not reject in Step~\ref{step:test-small-eps}, 
then it follows that $p(S) \geq 1-O(\eps/n^{1/4}) > 1-\eps/10$.
Therefore, $p|S$ should be at least $\eps/2$-far from $\U_S$ and 
the tester will output ``NO.'' This completes the proof of correctness.
\end{proof}

\section{Sample Complexity Lower Bound} \label{sec:lb}

In this section, we prove a sample size lower bound matching our algorithm \textsc{Gen-Uniformity-Test}. 
One part of the lower bound is fairly easy. In particular, it is known~\cite{Paninski:08}
that $\Omega(\sqrt{n}/\eps^2)$ samples are required to test uniformity of a distribution with a known support of size $n$. 
It is easy to see that the hard cases for this lower bound still work when $\|p\|_2 = \Theta(n^{-1/2})$. 

The other half of the lower bound is somewhat more difficult and we rely on the lower bound techniques of \cite{DK16}.
In particular, for $n>0,$ and $1/10>\eps>n^{-1/4}$ and for $N$ sufficiently large, we produce a pair of distributions 
$\mathcal{D}$ and $\mathcal{D'}$ over positive measures on $[N]$, so that:
\begin{enumerate}
\item A random sample from $\mathcal{D}$ or $\mathcal{D'}$ has total mass $\Theta(1)$ with high probability.
\item A random sample from $\mathcal{D}$ or $\mathcal{D'}$ has support of size $\Theta(n)$ with high probability.
\item \label{property:comp} A sample from $\mu\in\mathcal{D}$ has $\mu/\|\mu\|_1$ be the uniform distribution over some subset 
of $[N]$ with probability $1$.
\item \label{property:sound}  A sample from $\mu\in\mathcal{D'}$ has $\mu/||\mu\|_1$ 
be at least $\Omega(\eps)$-far from any uniform distribution with high probability.
\item \label{property:lb} Given a measure $\mu$ taking randomly from either $\mathcal{D}$ or $\mathcal{D'}$, 
no algorithm given the output of a Poisson process with intensity $k\mu$ for $k=o(\min(n^{2/3}/\eps^{4/3},n))$ 
can reliably distinguish between a $\mu$ taken from $\mathcal{D}$ and $\mu$ taken from $\mathcal{D'}$.
\end{enumerate}

Before we exhibit these families, we first discuss why the above is sufficient. 
This Poissonization technique has been used previously in various settings~\cite{VV14, DK16, WY16, DGPP17}, 
so we only provide a sketch here.
In particular, suppose that we have such families $\mathcal{D}$ and $\mathcal{D'}$, 
but that there is also an algorithm $A$ that distinguishes between a distribution $p$ being uniform 
and being $\eps$-far from uniform in $m=o(\eps^{-4/3}/\|p\|_3)$ samples. We show that we can use algorithm $A$ 
to violate property~\ref{property:lb} above. In particular, letting $p=\mu/\|\mu\|_1$ for $\mu$ a random measure taken 
from either $\mathcal{D}$ or $\mathcal{D'}$, we note that with high probability $p$ has support of size $\Theta(n)$, 
and thus $\|p\|_3 = O(n^{-2/3}).$ Therefore, $m'=o(n^{2/3}/\eps^{4/3})$ samples are sufficient to distinguish 
between $p$ being uniform and being $\Omega(\eps)$ far from uniform. However, by properties~\ref{property:comp} 
and~\ref{property:sound}, 
this is equivalent to distinguish between $\mu$ being taken from $\mathcal{D}$ and being taken from $\mathcal{D'}$. 
On the other hand, given the output of a Poisson process with intensity $Cm'\mu$, 
for $C$ a sufficiently large constant, a random $m'$ of these samples (note that there are at least $m'$ total samples with high probability) 
are distributed identically to $m'$ samples from $p$. Thus, applying $A$ to these samples distinguishes between 
$\mu$ taken from $\mathcal{D}$ and $\mu$ taken from $\mathcal{D'}$, thus contradicting property~\ref{property:lb}.

We now exhibit the families $\mathcal{D}$ and $\mathcal{D'}$. 
In both cases, we want to arrange $\mu_i:=\mu(\{i\})$ to be i.i.d. for different $i$. 
We also want it to be the case that the first and second moments of $\mu_i$ 
are the same for $\mathcal{D}$ and $\mathcal{D'}$. 
Combining this with requirements on closeness to uniform, we are led to the following definitions:

\noindent For $\mu$ taken from $\mathcal{D'}$, we let
$$
\mu_i = \begin{cases}\frac{1+\eps}{n} & \textrm{, with probability } \frac{n}{2N}\\ \frac{1-\eps}{n} & \textrm{, with probability } \frac{n}{2N}\\ 0 & \textrm{, otherwise\;.}\end{cases}
$$ 
For $\mu$ taken from $\mathcal{D}$, we let
$$
\mu_i = \begin{cases}\frac{1+\eps^2}{n} & \textrm{, with probability } \frac{n}{N(1+\eps^2)}\\ 0 & \textrm{, otherwise\;.}\end{cases}
$$ 

Note that in both cases, the average total mass is $1$, 
and it is easy to see by Chernoff bounds that the actual mass of $\mu$ is $\Theta(1)$ with high probability. 
Additionally, the support size is always $\Theta(n)$ times the total mass, and so is $\Theta(n)$ with high probability. 
For $\mu$ taken from $\mathcal{D}$, all of the $\mu_i$ are either $0$ or $\frac{1+\eps^2}{n}$, 
and thus $\mu/\|\mu\|_1$ is uniform over its support. For $\mu$ taken from $\mathcal{D'}$, 
with high probability at least a third of the bins in its support have $\mu_i = \frac{1+\eps}{n}$, 
and at least a third have $\mu_i = \frac{1-\eps}{n}$. If this is the case, then at least a constant fraction of the mass 
of $\mu/\|\mu\|_1$ comes from bins with mass off from the average mass by at least a $(1\pm \eps)$ factor, 
and this implies that $\mu/\|\mu\|_1$ is at least $\Omega(\eps)$-far from uniform.

We have thus verified 1-4. Property~\ref{property:lb} will be somewhat more difficult to prove. 
For this, let $X$ be a random $\{0,1\}$ random variable with equal probabilities. 
Let $\mu$ be chosen randomly from $\mathcal{D}$ if $X=0$, 
and randomly from $\mathcal{D'}$ if $X=1$. 
Let our Poisson process with intensity $k\mu$ return $A_i$ samples from bin $i$. 
We note that, by the same arguments as in~\cite{DK16}, 
it suffices to show that the shared information $I(X;A_1,\ldots,A_N)=o(1).$ 
In order to prove this, we note that the $A_i$ are conditionally independent on $X$, 
and thus we have that $I(X;A_1,\ldots,A_N) \leq \sum_{i=1}^N I(X;A_i) = NI(X;A_1)$. 
Thus, we need to show that $I(X;A_1)=o(1/N)$. 
For notational simplicity, we drop the subscript in $A_1$.

This boils down to an elementary but tedious calculation. 
We begin by noting that we can bound
$$
I(X;A) = \sum_{t=0}^\infty O\left(\frac{(\Pr(A=t|X=0)-\Pr(A=t|X=1))^2}{\Pr(A=t)} \right) \;.
$$
(This calculation is standard. See Fact~81 in~\cite{CDKS17} for a proof.)
We seek to bound each of these terms. 
The distribution of $A$ conditioned on $\mu_1$ is Poisson with parameter $k \mu_1$. 
Thus, the distribution of $A$ conditioned on $X$ is a mixture of two or three Poisson distributions, 
one of which is the trivial constant $0$. We start by giving explicit expressions for these probabilities.

Firstly, for the $t=0$ term, note that
$$\Pr(A=t|X=1) = 1-\frac{n}{N}\left(1-\frac{e^{-k(1+\eps)/n}+ e^{-k(1-\eps)/n}}{2}\right) \; ,$$
$$\Pr(A=t|X=0) = 1- \frac{n}{N(1+\eps^2)}(1- e^{-k(1+\eps^2)/n}) \; .$$
Note that $\Pr(A=0)$  is at least $1-n/N \geq 1/2$ and $\Pr(A=t|X=1)-\Pr(A=t|X=0) \leq n/N$. 
Thus, the contribution from this term, $\frac{(\Pr(A=0|X=0)-\Pr(A=0|X=1))^2}{\Pr(A=0)}$,
is $O(n/N)^2 = o(1/N)$.

For $ t \geq 1$, there is no contribution from $\mu_1=0$. 
We can compute the probabilities involved exactly as
$$\Pr(A=t|X=1) = \frac{n}{N}\frac{(k(1+\eps)/n)^t e^{-k(1+\eps)/n}+(k(1-\eps)/n)^t e^{-k(1-\eps)/n}}{2 t!} \;,$$
$$\Pr(A=t|X=0) = \frac{n}{N(1+\eps^2)} \frac{(k(1+\eps^2)/n)^t e^{-k(1+\eps^2)/n}}{t!} \;,$$
and obtain that $\frac{(\Pr(A=t|X=0)-\Pr(A=t|X=1))^2}{\Pr(A=t)}$ is
$$
O\left(\left(\frac{n^{1-t} k^t}{2Nt!} \right)
\frac
{\left((1+\eps)^t e^{-k(1+\eps)/n}+(1-\eps)^t e^{-k(1-\eps)/n} - 2(1+\eps^2)^{t-1} e^{-k(1+\eps^2)/n} \right)^2 }
{(1+\eps)^t e^{-k(1+\eps)/n}+(1-\eps)^t e^{-k(1-\eps)/n} + 2 (1+\eps^2)^{t-1} e^{-k(1+\eps^2)/n} }
\right) \;.
$$ 
Factoring out the $e^{-k/n}$ terms and noting that, since $k\eps/n=o(1)$, 
the denominator is $\Omega(e^{-k/n})$ yields that
$$
O\left(\left(\frac{n^{1-t} k^t e^{-k/n}}{2Nt!} \right)
\left((1+\eps)^t e^{-k(1+\eps)/n}+(1-\eps)^t e^{-k(1-\eps)/n} - 2(1+\eps^2)^{t-1} e^{-k(1+\eps^2)/n} \right)^2 
\right) \;.
$$ 
Noting that $k/n=o(1)$, we can ignore this $e^{-kn}$ term and Taylor expanding the exponentials, we have that 
\begin{align*}
& \frac{(\Pr(A=t|X=0)-\Pr(A=t|X=1))^2}{\Pr(A=t)} = \\
& O\left(\left(\frac{n^{1-t} k^t }{2Nt!} \right) \left((1+\eps)^t (1-k(1+\eps)/n)+(1-\eps)^t (1+k(1-\eps)/n)  \right. \right. \\
& -   \left. \left. 2(1+\eps^2)^{t-1} (1-k(1+\eps^2)/n) + O((k\eps/n)^2 (1+\eps)^t) \right)^2 \right) \;.
\end{align*}
We deal separately with the cases $t=1, t=2$ and $t > 2$.
For the $t=1$ term, we have
\begin{align*}
& O\left(\left(\frac{k}{N} \right) \left((1+\eps)(1-k\eps/n)+(1-\eps)(1+k\eps/n) - 2 (1-k\eps^2/n) +O((k\eps/n)^2)\right)^2 \right) \\
=& O\left(\left(\frac{k}{N} \right) O((k\eps/n)^2)^2 \right) \;.
\end{align*}
Since $k=o(n^{2/3}/\eps^{4/3})$ and $\eps > n^{-1/4}$, 
$\eps k/n = o(n^{-1/3}/\eps^{1/3}) = o(n^{-1/4})$, 
and we find that this is
\begin{align*}
& O\left(\left(\frac{k}{N} \right)o(1/n) \right)= o(1/N) \;.
\end{align*}
This appropriately bounds the contribution from this term.

When $t=2$, we have
\begin{align*}
& O\left(\left(\frac{k^2}{nN} \right) \left((1+\eps)^2 (1-k(1+\eps)/n)+(1-\eps)^2 (1-k(1-\eps)/n) \right. \right. \\
& \left. \left. - 2(1+\eps^2) (1-k(1+\eps^2)/n) + O((k\eps/n)^2) \right)^2 
\right) \;.
\end{align*}
Note that the terms without $k/n$ factors cancel out, $(1+\eps)^2 + (1-\eps)^2 - 2(1+\eps^2)=0$, yielding
$$
O(k^2/nN)(k\eps^2/n+o(n^{-1/2}))^2 = O(k^4\eps^4/n^3 N) + o(k^2/n^2 N) = o(k^3\eps^4/n^2 N) + o(1/N) = o(1/N) \;,
$$
using both $k=o(n^{2/3}/\eps^{4/3})$ and $k=o(n)$.

For $t>2$, we let $f_t(x)=(1+x)^t(1-kx/n)$.
In terms of $f_t$, we have that $\frac{(\Pr(A=t|X=0)-\Pr(A=t|X=1))^2}{\Pr(A=t)}$ is:
$$O\left(\left(\frac{n^{1-t} k^t }{2Nt!} \right)(f_t(\eps)+f_t(-\eps))/2-f_t(0)-(f_{t-1}(\eps^2)-f_{t-1}(0))+o(n^{-1/2})^2\right) \;.
$$
Using the Taylor expansion of $f_t$ in terms of its first two derivatives and $f_{t-1}$ in terms of its first, we see that
$$(f_t(\eps)+f_t(-\eps))/2-f_t(0) = \eps^2 f''_t(\xi)$$ 
and 
$$f_{t-1}(\eps^2)-f_{t-1}(0) = \eps^2 f'_{t-1}(\xi') \;,$$ 
for some $\xi \in [-\eps,\eps]$ and $\xi' \in [0,\eps^2]$.
However, the derivatives are 
$$f'_t(x)=(1+x)^{t-1}(t-(1+x+tx)k/n)$$ 
and $$f''_t(x)=(1+x)^{t-2}(t(t-1) - t(t+1)xk/n) \;,$$
and so $|f''_t(\xi)| \leq O(t^2(1+\eps)^{t-1})$ and $f'_{t-1}(\xi') \leq O(t(1+\eps^2)^{t-2})$.
Hence, the term 
$$\frac{(\Pr(A=t|X=0)-\Pr(A=t|X=1))^2}{\Pr(A=t)}$$ is at most
\begin{align*}
& O(n^{1-t}k^t/Nt!) (\eps^4 t^4 (1+\eps)^{2t-2})+o(1/n)) \\ 
& = O\left( (k^3\eps^4/n^2) (t^4(1+\eps)^2/N) (k(1+\eps)^2/n)^{t-3} / t! \right) + o\left((k/n)^{t}/(Nt!)\right)\\
& = o(1/N)t^4/t!  \;,
\end{align*}
using both $k=o(n^{2/3}/\eps^{4/3})$ and $k=o(n)$.
Since $(t+1)^4/(t+1)! \leq t^4/2t!$ for all $t \geq 4$,
even summing the above over all $t\geq 3$ still leaves $o(1/N)$.

Thus, we have that $I(X;A)=o(1/N)$, and therefore that $I(X:A_1,\ldots,A_N)=o(1)$. 
This proves that $X=0$ and $X=1$ cannot be reliably distinguished 
given $A_1,\ldots,A_N$, and thus proves property~\ref{property:lb}, 
completing the proof of our lower bound.

\section{Conclusions}
In this paper, we gave tight upper and lower bounds on the sample complexity 
of generalized uniformity testing -- a natural non-trivial generalization of uniformity testing, 
recently introduced in~\cite{BC17}.
The obvious research question is to understand the sample complexity of testing more general symmetric
properties (e.g., closeness, independence, etc.) for the regime where the domain of the underlying
distributions is discrete but unknown (of unknown size).
Is it possible to obtain sub-learning sample complexities for these problems?
And what is the optimal sample complexity for each of these tasks?
It turns out that the answer to the first question is affirmative.
These extensions require more sophisticated techniques and will appear in a forthcoming work.

\bibliographystyle{alpha}

\bibliography{allrefs}

\end{document}